\documentclass[12pt]{article}

\usepackage{amsmath}
\usepackage{amssymb,amsfonts,textcomp}

\usepackage{hyperref}
\hypersetup{pdftex, colorlinks=true, linkcolor=blue, citecolor=blue, filecolor=blue, urlcolor=blue, pdftitle=}

\usepackage{graphicx}
\usepackage{nicefrac}

\usepackage{vmargin}
\setpapersize{USletter}
\setmargnohfrb{1.0in}{1.0in}{1.0in}{1.0in}

\usepackage{amsthm}
\newtheorem{theorem}{Theorem}
\newtheorem{lemma}[theorem]{Lemma}
\newtheorem{corollary}[theorem]{Corollary}

\usepackage[footnotesize]{caption}

\usepackage{pxfonts}

\title{Weighted genomic distance can hardly impose a bound on the proportion of transpositions}

\author{
Shuai Jiang
\and
Max A. Alekseyev\footnote{Department of Computer Science and Engineering, University of South Carolina, Columbia, SC, USA. Email: {\tt maxal@cse.sc.edu}}
}

\begin{document}
\maketitle
\thispagestyle{empty}

\begin{abstract}
\emph{Genomic distance} between two genomes, i.e., the smallest number of genome rearrangements required to transform one genome into the other, 
is often used as a measure of evolutionary closeness of the genomes in comparative genomics studies. 
However, in models that include rearrangements of significantly different ``power'' such as \emph{reversals} (that are ``weak'' and most frequent rearrangements) and
\emph{transpositions} (that are more ``powerful'' but rare), the genomic distance typically corresponds to a transformation
with a large proportion of transpositions, which is not biologically adequate. 

\emph{Weighted genomic distance} is a traditional approach to bounding the proportion of transpositions by 
assigning them a relative weight $\alpha>1$. 
A number of previous studies addressed the problem of computing weighted genomic distance with $\alpha\leq 2$.

Employing the model of multi-break rearrangements on circular genomes, that captures both reversals (modelled as \emph{2-breaks}) 
and transpositions (modelled as \emph{3-breaks}), 
we prove that for $\alpha\in(1,2]$, a minimum-weight transformation may entirely consist of transpositions,
implying that the corresponding weighted genomic distance does not actually achieve its purpose of bounding the proportion of transpositions.
We further prove that for $\alpha\in(1,2)$, the minimum-weight transformations do not depend on 
a particular choice of $\alpha$ from this interval. We give a complete characterization of such transformations
and show that they coincide with the transformations that 
at the same time have the shortest length and make the smallest number of breakages in the genomes.

Our results also provide a theoretical foundation for the empirical observation that for $\alpha < 2$, 
transpositions are favored over reversals in the minimum-weight transformations.
\end{abstract}

\section{Introduction}

Genome rearrangements are evolutionary events that change genomic architectures. 
Most frequent rearrangements are \emph{reversals} (also called \emph{inversions}) that ``flip'' continuous segments within single chromosomes.
Other common types of rearrangements are \emph{translocations} that ``exchange'' segments from different chromosomes and \emph{fission/fusion} 
that respectively ``cut''/``glue'' chromosomes.

Since large-scale rearrangements happen rarely and have dramatic effect on the genomes, the number of rearrangements (\emph{genomic distance}\footnote{We 
remark that the term \emph{genomic distance} 
sometimes is used to refer to a particular distance under reversals, translocations, fissions, and fusions.}) between two genomes
represents a good measure for their evolutionary remoteness and often is used as such in phylogenomic studies.
Depending on the model of rearrangements, there exist different types of genomic distance~\cite{Fertin2009}.

Particularly famous examples are the \emph{reversal distance} between unichromosomal genomes~\cite{HP1} 
and the genomic distance between multichromosomal genomes under all aforementioned types of rearrangements~\cite{HP2}.
Despite that both these distances can be computed in polynomial time, their analysis is somewhat complicated, thus
limiting their applicability in complex setups.
The situation becomes even worse when the chosen model includes more ``complex'' rearrangement operations such as \emph{transpositions} 
that cut off a segment of a chromosome and insert it into some other place in the genome. 
Computational complexity of most distances involving transpositions, including the \emph{transposition distance}, 
remains unknown~\cite{Radcliffe05,Bader2007,Elias2007}.
To overcome difficulties associated with the analysis of genomic distances
many researchers now use simpler models of multi-break~\cite{AlekseyevTCS2008}, DCJ~\cite{Yancopoulos05}, block-interchange~\cite{Christie1996} rearrangements 
as well as \emph{circular} instead of \emph{linear} genomes, which give reasonable approximation to original genomic distances~\cite{AlekseyevJCB2008}.

Another obstacle in genomic distance-based approaches arises from the fact that transposition-like rearrangements 
are at the same time much rare and ``powerful'' than reversal-like rearrangements. 
As a result, in models that include both reversals and transpositions, 
the genomic distance typically corresponds to rearrangement scenarios with a large proportion of transpositions,
which is not biologically adequate. A traditional approach to bounding the proportion of transpositions is
\emph{weighted genomic distance} defined as the minimum weight of a transformation between two genomes, 
where transpositions are assigned a relative weight $\alpha>1$~\cite{Fertin2009}.
A number of previous studies addressed the weighted genomic distance for $\alpha\leq 2$.
In particular, 
Bader and Ohlebusch~\cite{Bader2007} developed a 1.5-approximation algorithm for $\alpha \in [1,2]$.
For $\alpha=2$, Eriksen~\cite{Eriksen2001} proposed a $(1+\epsilon)$-approximation algorithm 
(for any $\epsilon>0$).

Employing the model of multi-break rearrangements~\cite{AlekseyevTCS2008} on circular genomes, 
that captures both reversals (modelled as 2-breaks) and transpositions (modelled as 3-breaks), 
we prove that for $\alpha\in (1,2]$, a minimum-weight transformation may entirely consist of transpositions.
Therefore, the corresponding weighted genomic distance does not actually achieve its purpose of bounding the proportion of transpositions.
We further prove that for $\alpha\in(1,2)$, the minimum-weight transformations do not depend on a particular choice of $\alpha$ 
from this interval (thus are the same, say, for $\alpha=1.001$ and $\alpha=1.999$), 
and give a complete characterization of such transformations. 
In particular, we show that these transformations coincide with those that 
at the same time have the shortest length and make the smallest number of breakages in the genomes,
first introduced by Alekseyev and Pevzner~\cite{Alekseyev2007}.

Our results also provide a theoretical foundation for the empirical observation of Blanchette et al.~\cite{Blanchette1996}  
that for $\alpha < 2$, transpositions are favored over reversals in the minimum-weight transformations.

\section{Multi-break Rearrangements and Breakpoint Graphs}\label{Sec1}

We represent a circular chromosome on $n$ genes $x_1,x_2,\dots,x_n$ as a cycle graph on $2n$ edges alternating between directed ``obverse'' edges,
encoding genes and their directionality, and undirected ``black'' edges, connecting adjacent genes (Fig.~\ref{Fig2}a).
A genome consisting of $m$ chromosomes is then represented as $m$ such cycles.
The edges of each color form a perfect matching.

A \emph{$k$-break} rearrangement~\cite{AlekseyevTCS2008} is defined as replacement of a set of $k$ black edges 
in a genome with a different set of $k$ black edges forming matching on the same $2k$ vertices.
In the current study we consider only 2-break (representing reversals, translocations, fissions, fusions) and 3-break rearrangements (including transpositions).

For two genomes $P$ and $Q$ on the same set of genes,\footnote{From now on, we assume that given genomes are always one the same set of genes.} 
represented as black-obverse cycles and gray-obverse cycles respectively,
their superposition is called the \emph{breakpoint graph} $G(P,Q)$  \cite{Bafna1996}.
Hence, $G(P,Q)$ consists of edges of three colors (Fig.~\ref{Fig2}b):
directed ``obverse'' edges representing genes, undirected black edges representing adjacencies in the genome $P$,
and undirected gray edges representing adjacencies in the genome $Q$.
We ignore the obverse edges in the breakpoint graph and focus on the black and gray edges forming a collection of black-gray alternating cycles (Fig.~\ref{Fig2}c).

\begin{figure}[!t]
\centering \includegraphics[scale=0.4]{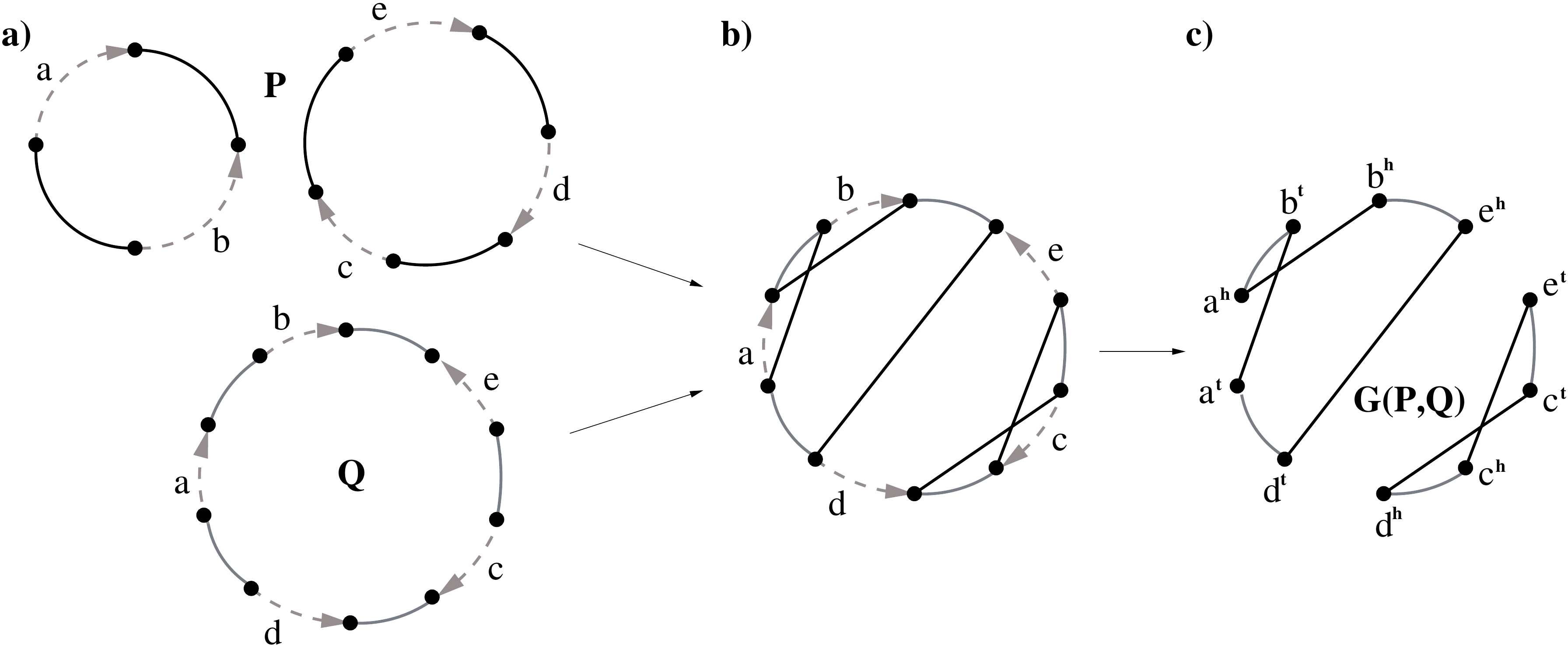}
\caption{{\bf a)} Graph representation of a two-chromosomal genome 
$P=(+a-b)(+c+e+d)$
as two black-obverse cycles and a unichromosomal genome $Q=(+a+b-e+c-d)$ as a gray-obverse cycle.
{\bf b)} The superposition of the genomes $P$ and $Q$.
{\bf c)} The breakpoint graph $G(P,Q)$ of the genomes $P$ and $Q$ (with removed obverse edges).
}
\label{Fig2}
\end{figure}

\begin{figure}[!t]
\centering \includegraphics[scale=0.4]{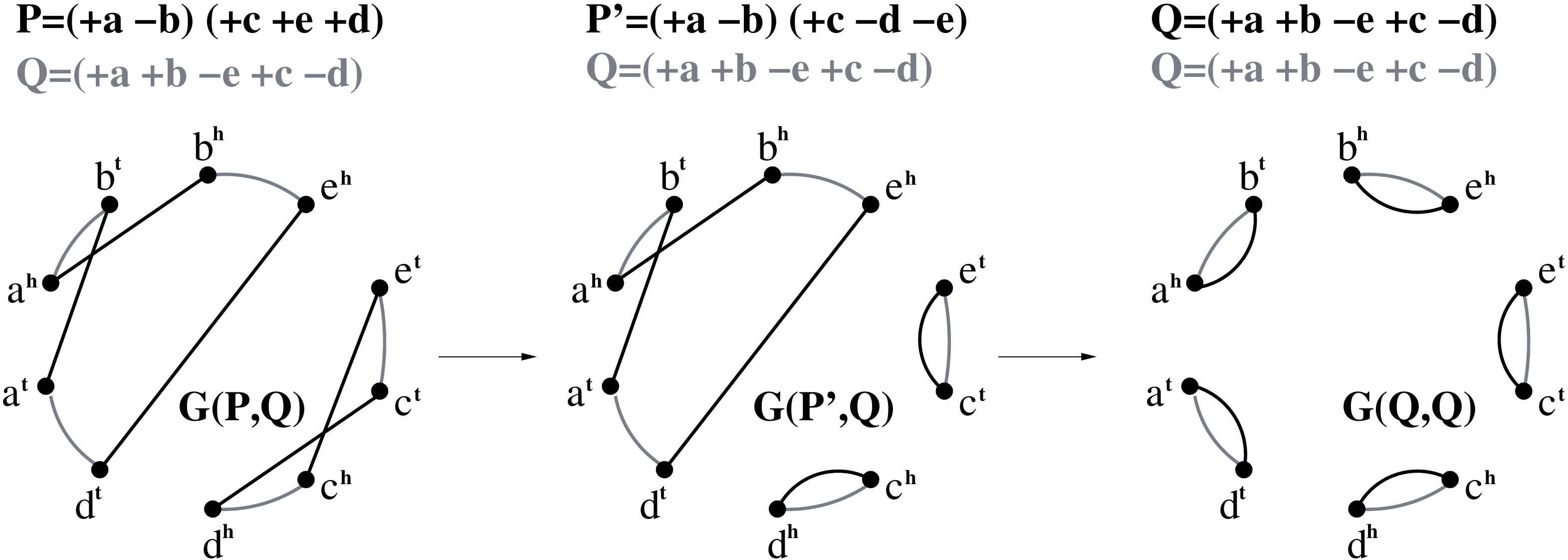}
\caption{A transformation between the genomes $P$ and $Q$ (defined in Fig.~\ref{Fig2}) 
and the corresponding transformation between the breakpoint graphs $G(P,Q)$ and $G(Q,Q)$
with a 2-break followed by a complete 3-break.}
\label{Fig3}
\end{figure}

A sequence of rearrangements transforming genome $P$ into genome $Q$ is called \emph{transformation}.
The length of a shortest transformation using $k$-breaks ($k=2$ or $3$) is called the \emph{$k$-break distance} between genomes $P$ and $Q$.

Any transformation of a genome $P$ into a genome $Q$ corresponds to a transformation of the breakpoint graph $G(P,Q)$ into
the \emph{identity breakpoint graph} $G(Q,Q)$ (Fig.~\ref{Fig3}). A close look at the increase in the number of black-gray cycles along this transformation,
allows one to obtain a formula for the distance between genomes $P$ and $Q$.
Namely, the 2-break distance is related to the number $c(P,Q)$ of black-gray cycles 
in $G(P,Q)$, while the 3-break distance is related to the number $c^{odd}(P,Q)$ of \emph{odd} black-gray cycles (i.e., black-gray cycles with an odd number of black edges):

\begin{theorem}[\cite{Yancopoulos05}]\label{Th2dist}
The 2-break distance between genomes $P$ and $Q$ is
$$d_2(P,Q) = |P| - c(P,Q).$$
\end{theorem}

\begin{theorem}[\cite{AlekseyevTCS2008}]\label{Th3dist}
The 3-break distance between genomes $P$ and $Q$ is
$$d_3(P,Q) = \frac{|P| - c^{odd}(P,Q)}2.$$
\end{theorem}

\section{Breakages and Optimal Transformations}\label{SecShortest}

Alekseyev and Pevzner~\cite{Alekseyev2007} studied the number of breakages\footnote{In \cite{Alekseyev2007}, the term \emph{break} 
is used. We use \emph{breakage} to avoid confusion with $k$-break rearrangements.} 
in transformations. The number of breakages made by a rearrangement is defined as the actual number of edges changed by this rearrangement. 
A 2-break always makes 2 breakages, while a 3-break can make 2 or 3 breakages. 
A 3-break making 3 breakages is called \emph{complete 3-break}. We treat non-complete 3-breaks as 2-breaks.

Alekseyev and Pevzner~\cite{Alekseyev2007} proved that between any two genomes, there always exists a transformation 
that simultaneously has the shortest length and makes the smallest number of breakages. 
We call such transformations \emph{optimal}.

For a 3-break $r$, 
we let $n_3(r)=1$ if $r$ makes 3 breakages (i.e., $r$ is a complete 3-break) and $n_3(r)=0$ otherwise. 
For a transformation $t$, we further define
$$n_2(t) = \sum_{r\in t} \left(1-n_3(r)\right)\qquad\text{and}\qquad n_3(t) = \sum_{r\in t} n_3(r)$$
that is, $n_2(t)$ and $n_3(t)$ are correspondingly the number of 2-breaks and complete 3-breaks in $t$.
If 2-breaks and complete 3-breaks are assigned respectively the weights $1$ and $\alpha$, then the weight of a transformation $t$ is
$$W_{\alpha}(t) = n_2(t) + \alpha\cdot n_3(t).$$

It is easy to see that a transformation $t$ has the length $n_2(t)+n_3(t)=W_1(t)$ and makes $2\cdot n_2(t) + 3\cdot n_3(t) = 2\cdot W_{\nicefrac32}(t)$ breakages overall.
Therefore, a transformation is optimal if and only if it simultaneously minimizes $W_1(t)$ and $W_{\nicefrac32}(t)$.
We generalize this result 
in Section~\ref{Sec2}
by showing 
that $\nicefrac{3}{2}$ can be replaced with any $\alpha\in (1,2)$.

For a rearrangement $r$ applied to a breakpoint graph, let $\Delta_r c^{odd}$ and $\Delta_r c^{even}$ be the resulting 
increase in the number of respectively odd and even black-gray cycles, respectively. 
Clearly, $\Delta_r c^{odd} + \Delta_r c^{even} = \Delta_r c$ gives the increase in the total number of black-gray cycles.

\begin{lemma}\label{Le0}
For any 3-break $r$, 
\begin{itemize}
\item $|\Delta_r c|\leq 1+n_3(r)$;
\item $\Delta_r c^{odd}$ is even and $|\Delta_r c^{odd}|\leq 2$;
\item $|\Delta_r c^{even}|\leq 1+n_3(r)$.
\end{itemize}
\end{lemma}

\begin{proof}
A 3-break $r$ operating on black edges in the breakpoint graph $G(P,Q)$ 
destroys at least one and at most three black-gray cycles. 
On the other hand, it creates at least one and at most three new black-gray cycles. 
Therefore, $|\Delta_r c| \leq 3-1 = 2$.
Similarly, if $n_3(r)=0$, then $|\Delta_r c| \leq 2-1 = 1$.

By similar arguments, we also have $|\Delta_r c^{odd}| \leq 3$ and $|\Delta_r c^{even}| \leq 3$.

Since the total number of black edges in destroyed and created black-gray cycles
is the same, $\Delta_r c^{odd}$ must be even. 
Combining this with $|\Delta_r c^{odd}|\leq 3$, we conclude that $|\Delta_r c^{odd}|\leq 2$.

If $\Delta_r c^{even}=3$, then the destroyed cycles must be odd, implying that $\Delta_r c^{odd}=-2$.
However, it is not possible for a 3-break to destroy two cycles and create three new cycles. Hence, $\Delta_r c^{even}\ne 3$.
Similarly, $\Delta_r c^{even}\ne -3$, implying that $|\Delta_r c^{even}|\leq 2$.
If $n_3(r)=0$ (i.e., $r$ is a 2-break), similar arguments imply $|\Delta_r c^{even}|\leq 1$.
\end{proof}

\begin{lemma}\label{Le1}
A transformation $t$ between two genomes is shortest if and only if $\Delta_r c^{odd}=2$ for every $r\in t$.
Furthermore, if $t$ is a shortest transformation between two genomes, then for every $r\in t$,
\begin{itemize}
\item if $n_3(r)=0$, then $\Delta_r c^{even}=-1$;
\item if $n_3(r)=1$, then $\Delta_r c^{even}=0$ or $-2$.
\end{itemize}
\end{lemma}

\begin{proof}
A transformation $t$ of a genome $P$ into a genome $Q$ increases the number of odd black-gray cycles 
from $c^{odd}(P,Q)$ in $G(P,Q)$ to $c^{odd}(Q,Q)=|P|$ in $G(Q,Q)$ with the total increase of 
$|P|-c^{odd}(P,Q)=2\cdot d_3(P,Q)$.
By Lemma~\ref{Le0}, $\Delta_r c^{odd}\leq 2$ for every $r\in t$ and thus
$$2\cdot d_3(P,Q) = \sum_{r\in t} \Delta_r c^{odd} \leq \sum_{r\in t} 2 = 2\cdot |t|,$$
implying that $|t|=d_3(P,Q)$ (i.e., $t$ is a shortest transformation) if and only if $\Delta_r c^{odd}=2$ for every $r\in t$.

Now let $t$ be a shortest transformation and thus $\Delta_r c^{odd}=2$ for every $r\in t$.
For a 2-break $r$ to have $\Delta_r c^{odd}=2$, it must be applied to an even black-gray cycle and split it into two odd black-gray cycles.
Thus any such $r$ also decreases the number of even black-gray cycles by 1, i.e., $\Delta_r c^{even}=-1$.

If a complete 3-break $r$ has $\Delta_r c^{odd}=2$, then $\Delta_r c^{even}=\Delta_r c - \Delta_r c^{odd} \leq 2-2=0$. 
By Lemma~\ref{Le0}, we also have $\Delta_r c^{even} \geq -2$ and $\Delta_r c^{even}\ne -1$, implying 
that $\Delta_r c^{even}=0$ or $-2$.
\end{proof}

By the definition, any optimal transformation is necessarily shortest. 
However, not every shortest transformation is optimal.
The following theorem characterizes optimal transformations within the shortest transformations:

\begin{theorem}\label{Th0}
A shortest transformation $t$ between two genomes is optimal if and only if for any $r\in t$, $\Delta_r c^{even}\ne -2$.
\end{theorem}

\begin{proof}
Let $t$ be a shortest transformation $t$ between two genomes. By Lemma~\ref{Le1}, $n_3(t)=u+v$ where $u$ is the number 
of complete 3-breaks with $\Delta_r c^{even}=0$ and $v$ is the number of complete 3-breaks with $\Delta_r c^{even}=-2$ (Fig.~\ref{Forbid2}).

\begin{figure}[!t]
\centering \includegraphics[scale=0.4]{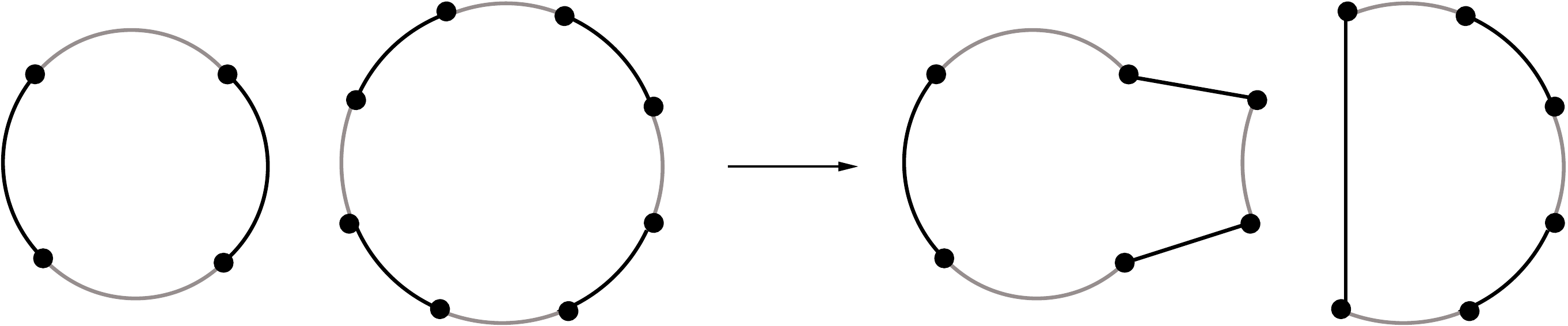}
\caption{A 3-break $r$ with $\Delta_r c^{odd} = 2$ and $\Delta_r c^{even} = -2$,
transforming two even black-gray cycles into two odd black-gray cycles. 
Such 3-breaks may appear in shortest transformations (Lemma~\ref{Le1}) but not in optimal ones (Theorem~\ref{Th0}).}
\label{Forbid2}
\end{figure}

With $n_2(t)$ 2-breaks and $n_3(t)=u+v$ complete 3-breaks $G(P,Q)$ is transformed into $G(Q,Q)$ with $|P|=|Q|$ trivial black-gray cycles, which are all odd.
By Lemma~\ref{Le1}, for the increase in the number of odd and even black-gray cycles in the breakpoint graph, we have:
$$
\begin{cases} 
c^{odd}(P,Q)+2(n_2(t)+u+v)=|P|, \\
c^{even}(P,Q)-n_2(t)-2v=0,
\end{cases}
$$
implying that 
\begin{eqnarray*}
W_{\nicefrac{3}{2}}(t) 
& = & n_2(t) + \frac{3}{2}(u+v) \\
& = & c^{even}(P,Q)-2v + \frac{3}{2}\left( \frac{|P|-c^{odd}(P,Q)}2 - c^{even}(P,Q)+2v\right) \\
& = & c^{even}(P,Q) + \frac{3}{2}\left( \frac{|P|-c^{odd}(P,Q)}2 - c^{even}(P,Q)\right) + v,
\end{eqnarray*}
which is minimal if and only if $v=0$, i.e., $\Delta_r c^{even}\ne -2$ for any $r\in t$.
\end{proof}

Lemma~\ref{Le1} and Theorem~\ref{Th0} imply:

\begin{corollary}\label{Cor0}
A transformation $t$ between two genomes is optimal if and only if for any $r\in t$,
\begin{itemize}
\item if $n_3(r)=0$, then $\Delta_r c^{odd}=2$ and $\Delta_r c^{even}=-1$;
\item if $n_3(r)=1$, then $\Delta_r c^{odd}=2$ and $\Delta_r c^{even}=0$.
\end{itemize}
\end{corollary}

\begin{theorem}\label{Th1}
A transformation $t$ between genomes $P$ and $Q$ is optimal if and only if
\begin{equation}\label{n23opt}
\begin{cases}
n_2(t) = c^{even}(P,Q), \\
n_3(t) = \frac{|P|-c^{odd}(P,Q)}2 - c^{even}(P,Q).
\end{cases}
\end{equation}
\end{theorem}

\begin{proof}
Let $t$ be an optimal transformation between genomes $P$ and $Q$. Then with $n_2(t)$ 2-breaks and $n_3(t)$ complete 3-breaks, 
it transforms $G(P,Q)$ into $G(Q,Q)$ with $|P|=|Q|$ trivial black-gray cycles, which are all odd.
By Corollary~\ref{Cor0}, we have
$$
\begin{cases} 
c^{odd}(P,Q)+2(n_2(t)+n_3(t))=|P|, \\
c^{even}(P,Q)-n_2(t)=0,
\end{cases}
$$
implying formulae \eqref{n23opt}.

Vice versa, a transformation $t$ between genomes $P$ and $Q$, satisfying \eqref{n23opt}, has the length $n_2(t)+n_3(t)=\frac{|P|-c^{odd}(P,Q)}2=d_3(P,Q)$, implying 
that $t$ is a shortest transformation. By Lemma~\ref{Le1}, $\Delta_r c^{even} = -1$ for every 2-break $r\in t$ and $\Delta_r c^{even} = 0$ or $-2$ for every complete 3-break $r\in t$.
Let $v$ be the number of complete 3-breaks $r\in t$ with $\Delta_r c^{even} = -2$. Then the increase in the number of even black-gray cycles along $t$ is
$$-c^{even}(P,Q) = - n_2(t) - 2v = - c^{even}(P,Q) - 2v,$$
implying that $v=0$ and thus $t$ is optimal by Theorem~\ref{Th0}.
\end{proof}

Theorem~\ref{Th1} implies that for some genomes, every optimal transformation consists entirely of complete 3-breaks:

\begin{corollary}\label{Cor1}
For genomes $P$ and $Q$ with $c^{even}(P,Q)=0$, 
every optimal transformation $t$ has $n_2(t)=0$ and thus consists entirely of complete 3-breaks.
\end{corollary}

\begin{corollary}\label{Cor2}
For an optimal transformation $t$ between genomes $P$ and $Q$,
$$W_{\alpha}(t) = c^{even}(P,Q) + \alpha\cdot \left(\frac{|P|-c^{odd}(P,Q)}2-c^{even}(P,Q)\right).$$
\end{corollary}

\section{Weighted multi-break distance}\label{Sec2}

Let $T(P,Q)$ be the set of all transformations between genomes $P$ and $Q$. 
For a real number $\alpha$, we define the weighted distance $D_{\alpha}(P,Q)$
between genomes $P$ and $Q$ as
$$D_{\alpha}(P,Q) = \min_{t\in T(P,Q)} W_{\alpha}(t)$$
that is, the minimum possible weight of a transformation between $P$ and $Q$.

Two important examples of the weighted distance are the ``unweighted'' distance $D_1(P,Q)=d_3(P,Q)$ and the distance
$D_{\nicefrac{3}{2}}(P,Q)$ equal the half of the minimum number of breakages in a transformation between genomes $P$ and $Q$.
By the definition of an optimal transformation, 
we have $D_{\nicefrac{3}{2}}(P,Q) = W_{\nicefrac{3}{2}}(t_0)$, where $t_0$ is an optimal transformation between genomes $P$ and $Q$.
Below we prove that $D_{\alpha}(P,Q) = W_{\alpha}(t_0)$ for any $\alpha\in (1,2]$.

\begin{theorem}\label{Th2}
For $\alpha \in (1,2]$, 
$$D_{\alpha}(P,Q) = W_{\alpha}(t_0),$$ 
where $t_0$ is any optimal transformation between genomes $P$ and $Q$.

Furthermore, for $\alpha \in (1,2)$, if $D_{\alpha}(P,Q) = W_{\alpha}(t)$
for a transformation $t$ between genomes $P$ and $Q$, then $t$ is an optimal transformation.
\end{theorem}

\begin{proof}
Let $t$ be any transformation and $t_0$ be any optimal transformation between genomes $P$ and $Q$.

We classify all possible changes in the number of even and odd 
black-gray cycles resulted from a single rearrangement $r$.
By Lemma~\ref{Le0}, $\Delta_r c^{odd}$ may take only values $-2,0,2$, 
while $|\Delta_r c| = |\Delta_r c^{odd} + \Delta_r c^{even}|\leq 1$
(if $r$ is a 2-break) or $\leq 2$ (if $r$ is a complete 3-break).
The table below lists the possible values of $\Delta_r c^{odd}$ and $\Delta_r c^{even}$,
satisfying these restrictions, 
along with the amount of rearrangements of each particular type in $t$, 
denoted $x_i$ for 2-breaks and $y_j$ for complete 3-breaks.

\begin{center}
\begin{tabular}{|c||c|c|c|c|c||c|c|c|c|c|c|c|c|c|c|c|}
\hline
& 
\multicolumn {5} {|c|} {$n_3(r)=0$} & \multicolumn {11} {|c|}{$n_3(r)=1$} \\
\hline
$\Delta_r c^{odd}$
&0 &0 &0 &-2 &2 & 0 &0 &0 &0 &0 &2 &2 &2 &-2 &-2 &-2\\
\hline
$\Delta_r c^{even}$
&0 &1 &-1 &1 &-1 &0 &1 &-1 &2 &-2 &0 &-1 &-2 &0 &1 &2\\
\hline
amount in $t$ & $x_1$ &$x_2$ &$x_3$ &$x_4$ & $x_5$ & $y_1$ &$y_2$ &$y_3$ &$y_4$ &$y_5$ &$y_6$ & $y_7$ &$y_8$ &$y_9$ & $y_{10}$ &$y_{11}$\\
\hline
\end{tabular}
\end{center}
For the transformation $t$, we have
$$
\begin{cases}
n_2(t) = x_1 + x_2 + x_3 + x_4 + x_5, \\
n_3(t) = y_1 + y_2 + y_3 + y_4 + y_5 + y_6 + y_7 + y_8 + y_9 + y_{10} + y_{11}.
\end{cases}
$$

Calculating the total increase in the number of odd and even black-gray cycles along $t$, we have
$$
\begin{cases}
-2x_4+2x_5+2y_6+2y_7+2y_8-2y_9-2y_{10}-2y_{11}=|P|-c^{odd}(P,Q),\\
x_2-x_3+x_4-x_5+y_2-y_3+2y_4-2y_5-y_7-2y_8+y_{10}+2y_{11}=-c^{even}(P,Q).
\end{cases}
$$
Theorem~\ref{Th1} further implies
$$
\begin{cases}
n_2(t_0) = -x_2+x_3-x_4+x_5-y_2+y_3-2y_4+2y_5+y_7+2y_8-y_{10}-2y_{11}, \\
n_3(t_0) = x_2-x_3+y_2-y_3+2y_4-2y_5+y_6-y_8-y_9+y_{11}.
\end{cases}
$$

Now we can evaluate the difference between the weights of $t$ and $t_0$ as follows:
\begin{eqnarray*}
W_\alpha (t)-W_\alpha (t_0) & = &  n_2(t) - n_2(t_0)+\alpha \cdot (n_3(t)-n_3(t_0)) \\
& = & x_1 + 2x_2+2x_4+y_2-y_3+2y_4-2y_5-y_7-2y_8+y_{10}+2y_{11} \\
&& +\ \alpha\cdot( -x_2+x_3+y_1+2y_3-y_4+3y_5+y_7+2y_8+2y_9+y_{10}) \\
& = & x_1 + (2-\alpha)\cdot x_2 + \alpha\cdot x_3 + 2x_4 + \alpha\cdot y_1+y_2+(2\alpha -1)\cdot y_3 +(2-\alpha)\cdot y_4 \\
&&+\ (3\alpha -2)\cdot y_5 + (\alpha-1)\cdot y_7 + (2\alpha -2)\cdot y_8 + 2\alpha \cdot y_9 + (\alpha+1)\cdot y_{10} + 2\cdot y_{11}.
\end{eqnarray*}
Since $\alpha \in (1,2]$ and $x_i, y_j \geq 0$, all summands in the last expression are nonnegative and thus 
$W_\alpha (t)-W_\alpha (t_0)\geq 0$. Since $t$ is an arbitrary transformation, we have
$$D_{\alpha}(P,Q) = W_{\alpha}(t_0).$$

For $\alpha\in (1,2)$, if $D_{\alpha}(P,Q) = W_{\alpha}(t)$ then 
$W_{\alpha}(t) - W_{\alpha}(t_0) = 0$, implying that only $x_5$ and $y_6$ 
(appearing with zero coefficients in the expression for $W_{\alpha}(t) - W_{\alpha}(t_0)$) 
can be nonzero and thus $t$ is optimal by Corollary~\ref{Cor0}.
\end{proof}

\section{Discussion}

We proved that for $\alpha\in (1,2]$, the minimum-weight transformations include the optimal transformations (Theorem~\ref{Th2})
that may entirely consist of transposition-like operations (modelled as complete 3-breaks) (Corollary~\ref{Cor1}).
Therefore, the corresponding weighted genomic distance does not actually impose any bound on the proportion of transpositions.

For $\alpha\in (1,2)$, we proved even a stronger result that the minimum-weight transformations coincide with the optimal transformations (Theorem~\ref{Th2}).
As a consequence we have that a particular choice of $\alpha\in (1,2)$ imposes no restrictions for the minimum-weight transformations
as compared to other values of $\alpha$ from this interval.
The value $\alpha=\nicefrac{3}{2}$ then proves that the optimal transformations coincide with 
those that 
at the same time have the shortest length and make the smallest number of breakages, 
studied by Alekseyev and Pevzner~\cite{Alekseyev2007}.
We further characterized the optimal transformations within the shortest transformations (i.e., the minimum-weight transformations for $\alpha=1$) 
by showing that the optimal transformations avoid one particular type of rearrangements (Theorem~\ref{Th0}, Fig.~\ref{Forbid2}).

It is worth to mention that the weighted genomic distance with $\alpha\geq 2$ is useless, 
since it allows (for $\alpha=2$) or even promotes (for $\alpha>2$) replacement of every complete 3-break with two equivalent 2-breaks, 
thus eliminating complete 3-breaks at all.

The extension of our results to the case of linear genomes will be published elsewhere.

\bibliographystyle{acm}
\bibliography{bioinf.bib}

\end{document}